\documentclass[ECP]{ejpecp} 

\usepackage{cancel} 
\usepackage{natbib}
\usepackage{graphicx}
\usepackage{subcaption}
\usepackage{cite}
\usepackage{mathrsfs}

\SHORTTITLE{SFS for general coalescents}

\TITLE{The site frequency spectrum for general coalescents}

\AUTHORS{Jeffrey P. Spence\footnote{University of California, Berkeley, USA. \EMAIL{spence.jeffrey@berkeley.edu}}{\ }$^{,}${\footnotemark[3]}, 
John A. Kamm\footnote{University of California, Berkeley, USA. \EMAIL{jkamm@stat.berkeley.edu}}{\ }$^{,}$\footnote{These authors contributed equally to this work.}, and 
Yun S. Song\footnote{University of California, Berkeley and University of Pennsylvania, USA. \EMAIL{yss@stat.berkeley.edu}}}


\AMSSUBJ{92D15} 
\AMSSUBJSECONDARY{65C50; 92D10} 

\SUBMITTED{2015}
\ACCEPTED{XXX} 

\VOLUME{XXX}
\YEAR{2015}
\PAPERNUM{XXX}
\DOI{vVOL-PID}
\DeclareMathOperator\Ei{Ei}
\DeclareMathOperator\diag{diag}
\DeclareMathOperator\textBeta{Beta}

\def\tmrca#1{T^{\text{\scriptsize MRCA}}_{#1}}

\def\Eip#1{\Ei\Big(#1\Big)}
\def\bfmath#1{\boldsymbol{#1}}
\newcommand{\Eone}{\text{E}_1}

\def\bfa{{\bfmath{a}}}

\def\bfc{{\bfmath{c}}}

\def\bfx{{\bfmath{x}}}

\def\bbE{\mathbb{E}}

\def\sS{\mathscr{S}}

\def\bftau{\boldsymbol{\tau}}
\def\bfgamma{\boldsymbol{\gamma}}
\def\branchLengths{\widetilde{\tau}}

\allowdisplaybreaks

\ABSTRACT{General genealogical processes such as $\Lambda$- and $\Xi$-coalescents, which respectively model multiple and simultaneous mergers, have important applications in studying marine species, strong positive selection, recurrent selective sweeps, strong bottlenecks, large sample sizes, and so on.  
Recently, there has been significant progress in developing useful inference tools for such general models.  In particular, 
inference methods based on the site frequency spectrum (SFS) have received noticeable attention.
Here, we derive a new formula for the expected SFS for general $\Lambda$- and $\Xi$-coalescents, which leads to an efficient algorithm.
For time-homogeneous coalescents, the runtime of our algorithm for computing the expected SFS is $O(n^2)$, where $n$ is the sample size. 
This is a factor of $n^2$ faster than the state-of-the-art method.  Furthermore, in contrast to existing methods, our method generalizes to time-inhomogeneous $\Lambda$- and $\Xi$-coalescents with measures that factorize as $\Lambda(dx)/\zeta(t)$ and $\Xi(d\bfx)/\zeta(t)$, respectively, where $\zeta$ denotes a strictly positive function of time.
The runtime of our algorithm in this setting is $O(n^3)$. 
We also obtain general theoretical results for the identifiability of the  $\Lambda$ measure  when $\zeta$ is a constant function, as well as for the identifiability of the  function $\zeta$ under a fixed $\Xi$ measure.}

\begin{document}

\section{Introduction}
When summarizing sequence data from $n$ individuals, a natural and often-used statistic is the site frequency spectrum (SFS), $\hat{\bftau}_n = (\hat{\tau}_{n,1},\ldots,\hat{\tau}_{n,n-1})^T$, where $\hat{\tau}_{n,k}$ is simply the number of sites at which $k$ out of $n$ individuals carry the mutant (or the derived) allele.  Despite being only $n-1$ numbers, the SFS still contains a surprising amount of information about the history and structure of the population from which the individuals were sampled.  Indeed, for neutrally evolving populations that are well-modeled by Kingman's coalescent \citep{kingman1982coalescent}, the expected value of the SFS was first computed for populations of constant size \citep{fu1995statistical}, extended to populations of variable size \citep{griffiths1998age,polanski2003note, polanski2003new}, and has since been used as a statistic for demographic inference in numerous studies (e.g. \citealt{gutenkunst2009inferring, bhaskar2015efficient, coventry2010deep,excoffier2013robust,gravel2011demographic,nielsen2000estimation,gao2015inference,kamm2015efficient}).

Yet, not all populations are well modeled by Kingman's coalescent.  In fact, Kingman's coalescent can be viewed as a special case of a broader class of coalescent processes called $\Lambda$-coalscents \citep{pitman1999coalescents, sagitov1999general}.  While Kingman's coalescent only permits pairwise mergers of lineages, $\Lambda$-coalescents allow two or more lineages to merge simultaneously in a single coalescence event.  Such events arise when a single individual has many offspring \citep{moehle2001classification,eldon2006coalescent}, under models of recurrent selective sweeps \citep{durrett2004approximating,durrett2005coalescent}, in populations undergoing continuous strong selection \citep{neher2013genealogies, schweinsberg2015rigorous}, and in many other models. $\Lambda$-coalescents can further be seen as special cases of a broader class of coalescents called $\Xi$-coalescents \citep{schweinsberg2000coalescents}.  In $\Xi$-coalescents, more than one merger event can occur simultaneously, resulting in simultaneous multiple mergers.  While $\Xi$-coalescents have received less attention than $\Lambda$-coalescents in the literature, they still arise in certain models of selection \citep{huillet2015pareto}, models of selective sweeps \citep{durrett2005coalescent}, models with repeated strong bottlenecks \citep{birkner2009modified}, and for certain diploid mating models \citep{moehle2003coalescent}.  Also, since $\Xi$-coalecents generalize $\Lambda$-coalescents, any results presented about $\Xi$-coalescents immediately pertain to $\Lambda$-coalescents.  

More formally, time-homogeneous $\Xi$-coalescents are governed by a measure $\Xi(d\bfx)$ on the set $\left\{(x_1,x_2,\ldots) : x_1 \ge x_2  \ge \cdots \ge 0,\sum_{i=1}^\infty x_i \le 1\right\}$.  Furthermore, we will consider time-inhomogeneous $\Xi$-coalescents with measures that decompose into a time-independent part $\Xi(d\bfx)$ and a strictly positive function $\zeta:\mathbb{R}_{\geq 0} \to \mathbb{R}_+$ of time, where $\zeta(t)$ represents (for historical reasons) the inverse intensity.  That is, the coalescent is now governed by the measure $\frac{\Xi(d\bfx)}{\zeta(t)}$.  For example, for Kingman's coalescent, $\Xi(d\bfx) = \delta_{\bold{0}}(d\bfx)$, the point mass at zero, and $\zeta(t)$ corresponds to the scaled effective population size at time $t$.  For other models, $\zeta(t)$ does not necessarily correspond to the population size, but has an interpretation specific to the model.  For example, \citet{neher2013genealogies} show empirically that the rate of coalescence events in a model of continuous strong selection is a nonlinear function of the population size and the first two moments of the distribution of mutational effects. For a review of the mechanics of $\Lambda$-coalescents, see \citet{pitman1999coalescents} and for a review of $\Xi$-coalescents, see \citet{schweinsberg2000coalescents}.  For an alternative perspective, see \citet{donnelly1999particle} and \citet{birkner2009modified} for a lookdown construction of particle systems with general reproduction mechanisms.

As mentioned above, the expected SFS for Kingman's coalescent is well understood, and can, in fact, be computed for an arbitrary $\zeta$ in $O(n^2)$ time \citep{polanski2003new}.  For $\Lambda$- and $\Xi$-coalescents, however, the expected SFS can only be computed for constant $\zeta$ and the method for $\Lambda$-coalescents takes $O(n^4)$ time \citep{birkner2013statistical} and the method for $\Xi$-coalescents takes time exponential in $n$ as a sum over partitions of the first $n$ numbers must be performed \citep{blath2015site}.  Here we present a method that can compute the expected SFS for time-inhomogeneous $\Lambda$- and $\Xi$-coalescents with arbitrary $\zeta$ in $O(n^3)$ time.
In the case where $\zeta$ is a constant function, our method can compute the expected SFS in $O(n^2)$ time given the rate matrix $\bold{Q}$ of the ancestral process, which will be defined more precisely below.  We also prove some results about the sample size needed to make $\Lambda$ identifiable for popular classes of $\Lambda$ measures for constant $\zeta$, as well as results about the sample size needed to make $\zeta$ identifiable for a fixed $\Xi(d\bfx)$.

There has also been some related work on determining the asymptotic behavior of the expected SFS as $n \to \infty$.  In this setting, 
\citet{berestycki2007beta,berestycki2014asymptotic} derive some simple formulae for time-homogeneous $\Lambda$-coalescents that come down from infinity.  For finite $n$, however, these asymptotic formulae can be rather inaccurate.  Indeed, even for $n = 10,000$, 
\citet{birkner2013statistical} show that for some $\Lambda$-coalescents, there is a sizable discrepancy between the asymptotic formulae and the SFS obtained by simulation, highlighting the need for finite-sample calculations.  Nevertheless, such asymptotic results highlight some interesting properties of $\Lambda$-coalescents and are reviewed in \citet{berestycki2009recent}.

The remainder of this paper is organized as follows.  We first present our main results about the computation of the SFS for time-inhomogeneous coalescents, and discuss the practical runtime of our implementation.  
We also investigate the variation in the empirical SFS and study the ability to infer the underlying model using the empirical SFS.  Then, we prove some identifiability results about general coalescents.  We conclude with discussion on the implications of our results.

%%%%%%%%%%%%%%%%%%%%%%%%%%%%%%%%%%%
\section{Main Theoretical Results on the Expected SFS} \label{sec:main}
%%%%%%%%%%%%%%%%%%%%%%%%%%%%%%%%%%%

Here we present our theoretical results on the expected SFS for a general $\Xi$-coalescent with a measure of the form $\frac{\Xi(d\bfx)}{\zeta(t)}$.  These results lead to an $O(n^3)$-time algorithm for computing the expected SFS and can be improved to $O(n^2)$ if $\zeta$ is a constant function.
Briefly, we use subsampling arguments to show that the expected SFS $\bftau_n = \bbE[\hat\bftau_n]$ can be computed from $\bfa_n := (\mathbb{E}\tmrca{2},\ldots,\mathbb{E}\tmrca{n})^T$, where $\mathbb{E}\tmrca{k}$ denotes the expected time to the most recent common ancestor for sample size $k\in\left\{2,\ldots,n\right\}$.  Then, we show how to compute $\bfa_n$ using a spectral decomposition of the rate matrix $\bold{Q}$ of the ancestral process (also known as the block-counting process) of the time-homogeneous coalescent corresponding to $\Xi(d\bfx)$.
More specifically, $\bold{Q}$ is a lower triangular matrix where $(\bold{Q})_{ij}$ is the instantaneous rate at which $i$ unlabeled lineages merge to form $j$ unlabeled lineages when $\zeta\equiv 1$.  For example, for Kingman's coalescent,
\[
(\bold{Q})_{ij} =  \begin{cases} 
      {i \choose 2}, & j = i-1, \\
      -{i \choose 2}, & j=i, \\
      0, & \text{otherwise.}
   \end{cases}
\]
Using this notation, we are now ready to state our main result.  The rest of the section will then provide lemmas which contain formulae for the matrices in Theorem~\ref{thm:main}, as well as a proof of those lemmas and Theorem~\ref{thm:main}.

\begin{theorem}
\label{thm:main}
Consider an arbitrary time-inhomogeneous $\Xi$-coalescent governed by a measure $\frac{\Xi(d\bfx)}{\zeta(t)}$, such that the expected time $c_{k,k}$ to the first coalescence for a sample of size $k$ is finite for $k\in\{2,\ldots,n\}$.  
Let $\bfc_n=(c_{2,2},\ldots,c_{n,n})^T$.  Then, there exists a universal matrix $\bold{A} \in \mathbb{R}^{n-1 \times n-1}$ that does not depend on the measure and a matrix $\bold{L} \in \mathbb{R}^{n-1 \times n-1}$ that depends on  $\Xi$ but not $\zeta$, such that
\begin{align*}
\bftau_n = \frac{\theta}{2}\bold{A}  \bfa_n \hspace{5mm}\text{and}\hspace{5mm}
\bfa_n = \bold{L}  \bfc_{n},
\end{align*}
where $\frac{\theta}{2}$ is the population-scaled mutation rate.  Furthermore, this allows $\bftau_n$ to be computed in $O(n^3)$ time.
\end{theorem}
Computing the matrix $\bold{L}$ in Theorem~\ref{thm:main} is costly.  For time-homogeneous coalescents, it is possible to compute $\bfa_n$ directly, resulting in the following corollary:
\begin{corollary}
\label{cor:constant}
In the same setting as Theorem~\ref{thm:main}, if $\zeta$ is a constant function, then $\bftau_n$ can be computed in $O(n^2)$ time.
\end{corollary}

In what follows, Lemmas~\ref{lem:anti} and \ref{lem:tmrca} provide formulae to compute the universal matrix $\bold{A}$, while Lemmas~\ref{lem:spectral} and \ref{lem:first} provide formulae to compute $\bold{L}$, which is related to the spectral decomposition of the rate matrix $\bold{Q}$.
The expected first coalescence times $\bfc_n=(c_{2,2},\ldots,c_{n,n})^T$ can be computed as \citep{polanski2003new,bhaskar2015efficient}
\begin{align*}
c_{k,k} = \int_{0}^{\infty} \mathbb{P}\left\{ \text{time of first coalescence for $k$ individuals} > t \right\} dt
= \int_0^\infty e^{(\bold{Q})_{kk} \int_0^t \frac{1}{\zeta(s)}ds} dt.
\end{align*}
Note that since $\bold{A}$ and $\bold{L}$ do not depend on $\zeta$, the SFS depends on time and the inhomogeneity of the coalescent process only through the first coalescence times $\bfc_n$.
%%%%%%%%%%%%%% BEGIN LEMMA 1
\begin{lemma}
\label{lem:anti}
Let $\bfgamma_n  := (\tau_{2,1},\tau_{3,2},\ldots, \tau_{n,n-1})^T$ denote the anti-singleton entries  (i.e., entries where exactly one individual has the ancestral allele and all other individuals have the derived allele) of the SFS for samples of sizes $2,\ldots,n$.  Then, 
\[
\bftau_n = \bold{B}  \bfgamma_n,
\]
where the entries of $\bold{B} \in \mathbb{R}^{n-1 \times n-1}$ are given by
\[
(\bold{B})_{ij} =
\begin{cases}
(-1)^{i-j} \frac{1}{j+1} {n-i-1 \choose j - i}   {n \choose i}, & i \le j,\\
0, & i>j.
\end{cases}
\]
\end{lemma}
%%%%%%%%%%%%%% END LEMMA 1
\begin{proof}
We use induction to show that
\begin{equation}
\tau_{n,i} =  \sum_{j=i}^{n-1}(-1)^{i-j} \frac{1}{j + 1} {n-i-1 \choose j - i}  {n \choose i}  \tau_{j+1,j}.
\label{eq:tau}
\end{equation}
Using exchangeability and a subsampling argument similar to that of 
\citet[Lemma 2]{kamm2015efficient},
we obtain, for $k > l+1$, 
\begin{equation}
\tau_{k-1,l} = \frac{l+1}{k}\tau_{k,l+1} + \frac{k - l}{k}\tau_{k, l},
\label{eq:tau_recursion}
\end{equation}
which follows from removing an individual uniformly at random from a sample of size $k$.
Now, define the \emph{level} of  $\tau_{n,i}$ as $n-i$ and note that \eqref{eq:tau} holds for level 1, i.e., for  $\tau_{l,l-1}$ on the left hand side.  
Assume that \eqref{eq:tau} holds for level $n-i-1$.  Then,
\begin{align*}
\tau_{n,i} &=  \frac{n}{n-i} \tau_{n-1,i} - \frac{i+1}{n-i} \tau_{n,i+1}\\
&= \frac{n}{n-i}\Bigg[\sum_{j=i}^{n-2} (-1)^{i - j } \frac{1}{j+1} {n-i -2\choose j - i}  {n-1 \choose i} \tau_{j+1,j} \Bigg] \\ 
&\text{\hspace{.3in}} - \frac{i+1}{n-i} \Bigg[ \sum_{j=i+1}^{n-1} (-1)^{i+1-j} \frac{1}{j+1} {n-i-2 \choose j - i-1}  {n \choose i+1}\tau_{j+1,j}\Bigg]\\
&= {n \choose i}  \Bigg\{\frac{1}{i+1} \tau_{i+1,i}   + (-1)^{n- 1 - i}  \frac{1}{n}  \tau_{n,n-1} 
\\ &\text{\hspace{1in}} +  \sum_{j=i+1}^{n-2} (-1)^{i - j} \frac{1}{j+1} \Bigg[{n-i-2 \choose j - i}  + {n-i-2 \choose j - i-1}\Bigg]\tau_{j+1,j}\Bigg\}\\
&= {n \choose i} \sum_{j=i}^{n-1} (-1)^{j-i} \frac{1}{j+1} {n-i -1 \choose j - i}  \tau_{j+1,j},
\end{align*}
where the first equality holds by the recursion \eqref{eq:tau_recursion} and the second equality holds by the inductive hypothesis, by noting that $\tau_{n-1,i}$ and $\tau_{n,i+1}$ are both one level below $\tau_{n, i}$.  
\end{proof}

The following lemma relates $\bfgamma_n$ to $\bfa_n$:

\begin{lemma}
\label{lem:tmrca}
Let $\bfgamma_n$, $\bfa_n$, and $\theta$ be defined as above. Then, 
\[
\bfgamma_n = \frac{\theta}{2}\bold{C}  \bfa_n,
\]
where $\bold{C} \in \mathbb{R}^{n-1 \times n-1}$ is bi-diagonal with $(\bold{C})_{k,k-1} =-(k+1)$ and $(\bold{C})_{kk} = k+1$ for $k \in \left\{2,\ldots,n-1\right\}$, and $(\bold{C})_{11} =2$.
\end{lemma}
\begin{proof}
As in the proof of Lemma~\ref{lem:anti}, we employ a subsampling argument.  Consider a sample of size $k+1$.  The only way that a subsample of size $k$ can have a different time to most recent common ancestor is if the removed individual is a singleton after all of the other lineages have coalesced.  The probability that we remove that singleton to form our subsample is $\frac{1}{k+1}$.  Then, the expected amount of time during which there is one singleton and all of the other individuals have coalesced scaled by the mutation rate is exactly the anti-singleton entry.  Thus,
\[
\frac{1}{k+1} \tau_{k+1, k} = \frac{\theta}{2} (\mathbb{E}\tmrca{k+1} - \mathbb{E}\tmrca{k})
\]
for $k>1$.  When $k=1$, there are only 2 lineages, so the total branch length is the anti-singleton entry.  Thus, $\tau_{2,1} = \frac{\theta}{2}2\mathbb{E}\tmrca{2}$.  Rewriting this as a matrix equation for $k\in\left\{1,\ldots,n-1\right\}$ completes the proof.
\end{proof}

By combining Lemmas~\ref{lem:anti} and \ref{lem:tmrca}, we obtain the universal matrix $\bold{A}=\bold{B}\bold{C}$.  We now show how to compute the $\Xi$-dependent matrix $\bold{L}$.  First, we establish the following result on the decomposition of the rate matrix $\bold{Q}$; this result was also obtained by 
\citet[Equation 2.3]{moehle2014spectral} for the Bolthausen-Sznitman coalescent.  
\begin{lemma}
\label{lem:spectral}
Fix an arbitrary $\Xi$-coalescent with $\lambda_i \ne \lambda_j$ for $i \ne j$, where $\lambda_i := \sum_{k=1}^{i-1} (\bold{Q})_{ik} = - (\bold{Q})_{ii}$.  Let $\bold{Q} \in \mathbb{R}^{n\times n}$ denote the rate matrix of the ancestral process corresponding to $\Xi(d\bfx)$ (that is the process counting the number of extant lineages at time $t$).  Then,
\[
\bold{Q} = \bold{U} \bold{E} \bold{U}^{-1},
\]
where  $(\bold{E})_{ij} = \delta_{ij} (\bold{Q})_{ii}$, with 
 $\delta_{ij}$ being the Kronecker delta which equals $1$ if $i=j$ and $0$ otherwise,  and 
\begin{align*}
(\bold{U})_{ij} &= 
\begin{cases}
1, & i=j,\\
\frac{1}{\lambda_i - \lambda_j} \sum_{k=j}^{i-1} (\bold{Q})_{ik} (\bold{U})_{kj},& i > j,\\
0, &\text{otherwise}.
\end{cases}
\end{align*}
\end{lemma}
\begin{proof}
By the construction of $\bold{U}$,
\[
(\bold{U})_{ij} (\bold{Q})_{jj} = \sum_{k=j}^{i} (\bold{Q})_{ik}(\bold{U})_{kj},
\]
which implies that $\bold{U}\bold{E} = \bold{Q}\bold{U}$.  Then, since $\bold{U}$ is triangular and has strictly positive diagonal entries, it is invertible.  Therefore, $\bold{Q} = \bold{U} \bold{E} \bold{U}^{-1}$.
\end{proof}

The following result relates $\bfa_n := (\mathbb{E}\tmrca{2},\ldots,\mathbb{E}\tmrca{n})^T$ and $\bfc_n=(c_{2,2},\ldots,c_{n,n})^T$:
\begin{lemma}
\label{lem:first}
Let $\bfa_n$ and $\bfc_n$ be defined as above.   Fix an arbitrary $\Xi$ measure and a strictly positive function $\zeta$.  Now consider a time-inhomogeneous coalescent governed by $\frac{\Xi(d\bfx)}{\zeta(t)}$.  If $c_{k,k} < \infty,$ for $2 \le k \le n$, then
\[
\bfa_n = -(\bold{U} \bold{D})_{2:n,2:n} \bfc_n,
\]
where $\bold{D} \in \mathbb{R}^{n \times n}$ is the diagonal matrix $\diag([\bold{U}^{-1}]_{\cdot,1})$, with 
$[\bold{U}^{-1}]_{\cdot,1}$ denoting the first column of $\bold{U}^{-1}$, and $(\bold{U} \bold{D})_{2:n,2:n}$ denotes the submatrix of $\bold{U} \bold{D}$ in rows and columns $2$ through $n$.
\end{lemma}
\begin{proof}
Note that $\mathbb{E}\tmrca{k} = \int_0^\infty \mathbb{P}\left\{ \tmrca{k} > t \right\} dt$.  Therefore,
\begin{align*}
\mathbb{E}\tmrca{k} = \int_0^\infty \mathbb{P}\left\{ \tmrca{k} > t \right\}dt &= \int_0^\infty  \sum_{l=2}^k \big[e^{\bold{Q} \int_0^t  \frac{1}{\zeta(s)} ds}\big]_{kl}  dt\\
&= \int_0^\infty   \sum_{l=2}^n \big[e^{\bold{Q} \int_0^t  \frac{1}{\zeta(s)} ds}\big]_{kl}  dt\\
&= \int_0^\infty \sum_{l=2}^n \big[\bold{U}e^{\bold{E}\int_0^t  \frac{1}{\zeta(s)} ds} \bold{U}^{-1}\big]_{kl} dt,
\end{align*}
where the third equality follows from the fact that $\bold{Q}$ is lower triangular and hence so is its exponential.  Now, since $\bold{U}$ is lower triangular, its inverse is as well.  Therefore, we may ignore the value of $[e^{\bold{E}\int_0^t  \frac{1}{\zeta(s)} ds} ]_{1,1}$.  Letting $\bold{F}(t) := e^{\bold{E}\int_0^t  \frac{1}{\zeta(s)} ds}$ but with $\bold{F}_{1,1}(t) := 0$, note that $\int_{0}^{\infty} \bold{F}(t) dt = \diag(0,\bfc_n)$.
 Then we have
\begin{align*}
\mathbb{E}\tmrca{k} = \int_0^\infty \sum_{l=2}^n [\bold{U}\bold{F}(t)\bold{U}^{-1}]_{kl} dt
&= \sum_{l=2}^n [\bold{U} \diag(0,\bfc_n) \bold{U}^{-1}]_{kl}.
\end{align*}
Now, note that $(\bold{U})_{i,1} = 1$ for all $i$ by Lemma~\ref{lem:spectral} and induction.  This implies 
$\sum_{l=1}^n [\bold{U}^{-1}]_{il} = \delta_{i1}$, or $\sum_{l=2}^n [\bold{U}^{-1}]_{il} = \delta_{i1} - [\bold{U}^{-1}]_{i1}$.
Using this identity, we can rewrite the above expression for $\mathbb{E}\tmrca{k}$ as 
\[
	\mathbb{E}\tmrca{k} = - \sum_{j=2}^n [ (\bold{UD})_{2:n,2:n}]_{k-1,j} c_{n,j},
\]
where $\bold{D}=\diag([\bold{U}^{-1}]_{\cdot,1})$.
Collecting these equations over $k\in \{2,\ldots,n\}$ in matrix form leads to the desired result.
\end{proof}
Using Lemma~\ref{lem:first}, we now see that the matrix $\bold{L}$ from Theorem~\ref{thm:main} is simply $-(\bold{UD})_{2:n,2:n}$.  Lemma~\ref{lem:spectral} provides a recursion to compute $\bold{U}$, and $\bold{D}$ may be computed by noting that $(\bold{U}^{-1})_{11} = 1$ and then since $\bold{U}\bold{U}^{-1} = \bold{I}$ we have
\[
\bold{U}^{-1}_{i1} = - \sum_{j = 1}^{i-1} (\bold{U})_{ij} (\bold{U}^{-1})_{j1}.
\]

\begin{proof}[Proof of Theorem~\ref{thm:main}]
Combining Lemmas~\ref{lem:anti}, \ref{lem:tmrca} and \ref{lem:first} we obtain the equations in the theorem.  For the runtime, note that each of the $O(n^2)$ entries of $\bold{U}$ requires $O(n)$ computations, and so computing $\bold{U}$ is $O(n^3)$.  The matrices composing $\bold{A}$ are known in closed form, however, and constructing $\bold{D}$ only requires filling $O(n)$ entries, each requiring $O(n)$ computations for a total of $O(n^2)$.  To then obtain the SFS from $\bfc_n$ simply requires iterated matrix vector products taking $O(n^2)$ time.  The overall procedure thus requires $O(n^3)$.
\end{proof}
\begin{lemma}
\label{lem:anconstant}
For coalescents of the form $\frac{\Xi(d\bfx)}{\zeta(t)}$ where $\zeta$ is a constant function, $\bfa_n$ can be computed recursively from $\bfc_n$ and $\bold{Q}$ as follows:
\begin{align*}
\bbE T_2^{TMRCA} &= c_{2,2}\\
\bbE T_k^{TMRCA} &= c_{k,k} + \sum_{l=2}^{k-1} \frac{(\bold{Q})_{kl}}{\lambda_k} \bbE T_l^{TMRCA}, & \text{for } k>2.
\end{align*}
\end{lemma}
\begin{proof}
The formulae follow immediately from the homogeneity of the process, recursing on the number of individuals, and noting that the probability that the first coalescence event for a sample of size $k$ results in $k$ lineages merging down to $l$ lineages is  $\frac{(\bold{Q})_{kl}}{\lambda_k}$.
\end{proof}
\begin{proof}[Proof of Corollary~\ref{cor:constant}]
Use Lemma~\ref{lem:anconstant} to compute $\bfa_n$ in $O(n^2)$ time.  Then, $\bftau_n = \bold{A}\bfa_n$ by Theorem~\ref{thm:main}, which also takes $O(n^2)$ time to compute.
\end{proof}
\begin{remark}
Other than computing $\bold{U}$, the algorithm presented in Theorem~\ref{thm:main} is $O(n^2)$.  Thus, for the Bolthausen-Sznitman Coalescent \citep{bolthausen1998ruelle} or Kingman's coalescent, where $\bold{U}$ is known in closed form \citep[Theorem~1.1 and Appendix]{moehle2014spectral}, the SFS can be computed in $O(n^2)$ time even for non-constant $\zeta$.
\end{remark}
\begin{remark}
The above results can easily be extended to a coalescent where both $\zeta$ and $\Xi(d\bfx)$ depend on $t$, so long as $\Xi(d\bfx)$ is piecewise constant.  For example, in the recent past the population may evolve according to a $\textBeta$-coalescent, whereas for $t$ greater than some $t_0$ the population may evolve according to Kingman's coalescent.  By setting $\zeta$ appropriately in Theorem~\ref{thm:main}, one may obtain a ``truncated SFS'' \citep{kamm2015efficient} for each different $\Xi(d\bfx)$.  Then, using the truncated SFS for each epoch and the same machinery as in \citet{kamm2015efficient} one may compute the full SFS.  The same techniques also allow one to consider multiple populations, with each population perhaps evolving according to its own $\Xi$ measure.
\end{remark}

%%%%%%%%%%%%%%%%%%%%%%%%%%%%%%%%%%%%
\section{Numerical Results}\label{sec:runtime}
%%%%%%%%%%%%%%%%%%%%%%%%%%%%%%%%%%%%

We implemented Theorem~\ref{thm:main} and Corollary~\ref{cor:constant} in Mathematica, and the notebook is available upon request.  We can compute the SFS for an arbitrary coalescent for a sample of size $n=100$ in approximately one second and a sample of size $n=300$ in a matter of minutes on a laptop computer, which is orders of magnitude faster than the more than one hour reported for a sample size of $n=100$ using the current state-of-the-art method \citep{blath2015site}.  Furthermore, 
\citet{blath2015site} only consider specific $\Xi$ measures where the number of simultaneous multiple mergers is restricted.  Our method has the same runtime for all $\Xi$ measures (after computing the rate matrix and the vector of first coalescence times).  See Figure~\ref{fig:runtime} for runtime versus sample size.  Furthermore, as noted above, if the spectral decomposition of the rate matrix $\bold{Q}$ is known, then the algorithm is $O(n^2)$.  We also present runtimes for the Bolthausen-Sznitman coalescent (which has a closed form solution for the spectral decomposition \citep{moehle2014spectral}) in Figure~\ref{fig:runtime}.  

\begin{figure}[t]
  \centering
\includegraphics[trim=0mm 5mm 0mm 15mm, clip, width=0.5\textwidth]{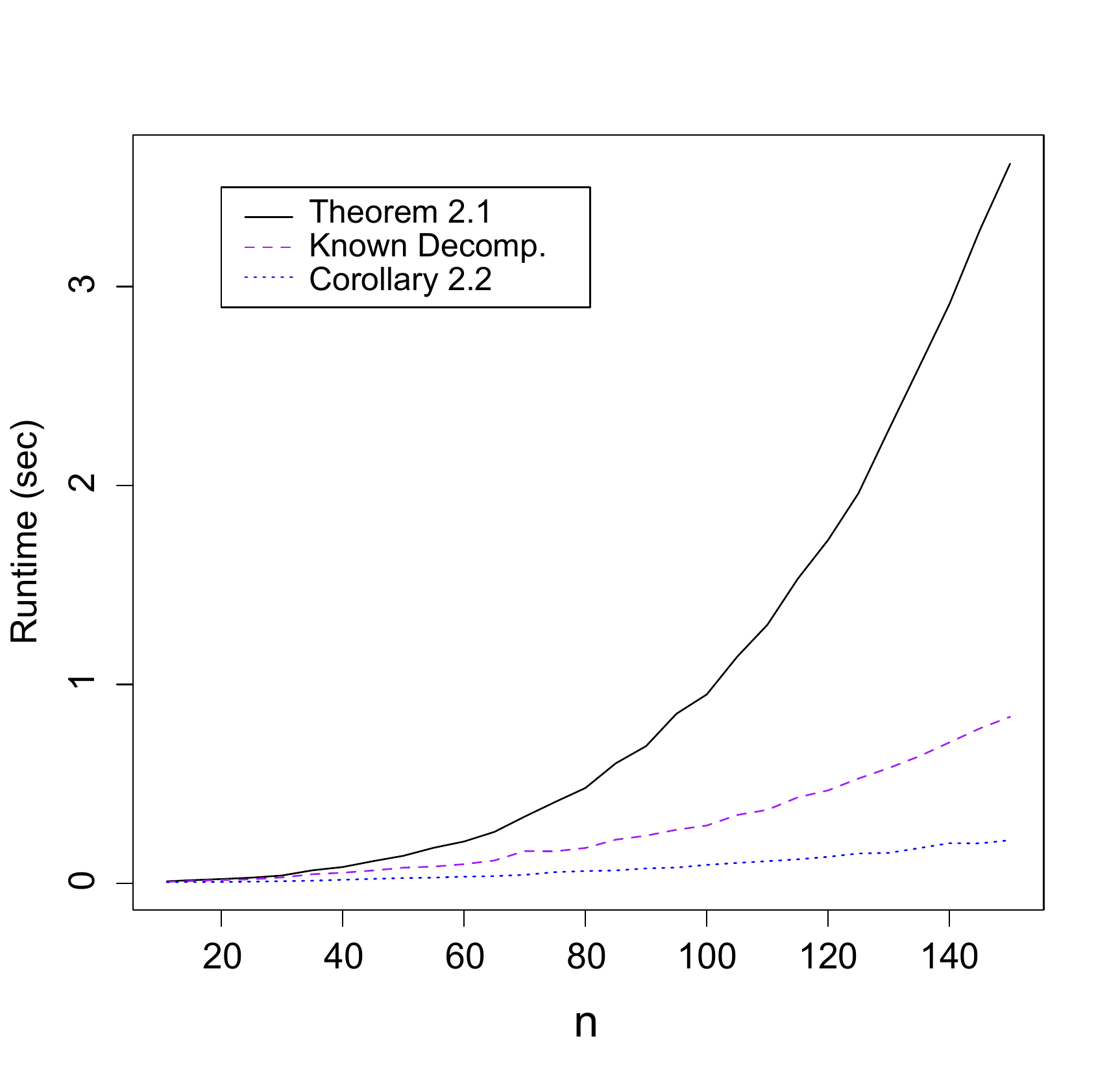}
\caption{Runtime result (in seconds).  Theorem~\ref{thm:main} was used to compute the SFS for the time-homogeneous Bolthausen-Sznitman coalescent.  The solid line uses Lemma~\ref{lem:spectral} to compute the spectral decomposition of $\bold{Q}$ resulting in a cubic runtime.  The dashed line uses the closed-form representation of the spectral decomposition of the Bolthausen-Sznitman coalescent \citep[Theorem~1.1]{moehle2014spectral} to compute the SFS in quadratic time.  The dotted line uses Corollary~\ref{cor:constant}, which is also quadratic.}
\label{fig:runtime}
\end{figure}

As long as the rate matrix $\bold{Q}$ of the ancestral process can be found exactly, our method is numerically stable.  This is the case for popular $\Lambda$-coalescents such as point-mass coalescents and $\textBeta$-coalescents, as well as point mass $\Xi$-coalescents.  If the rate matrix must be evaluated numerically, however, high precision computation may be needed to avoid potential numerical problems due to catastrophic cancellation.

\begin{figure}[t]
  \centering
\includegraphics[width=0.75\textwidth]{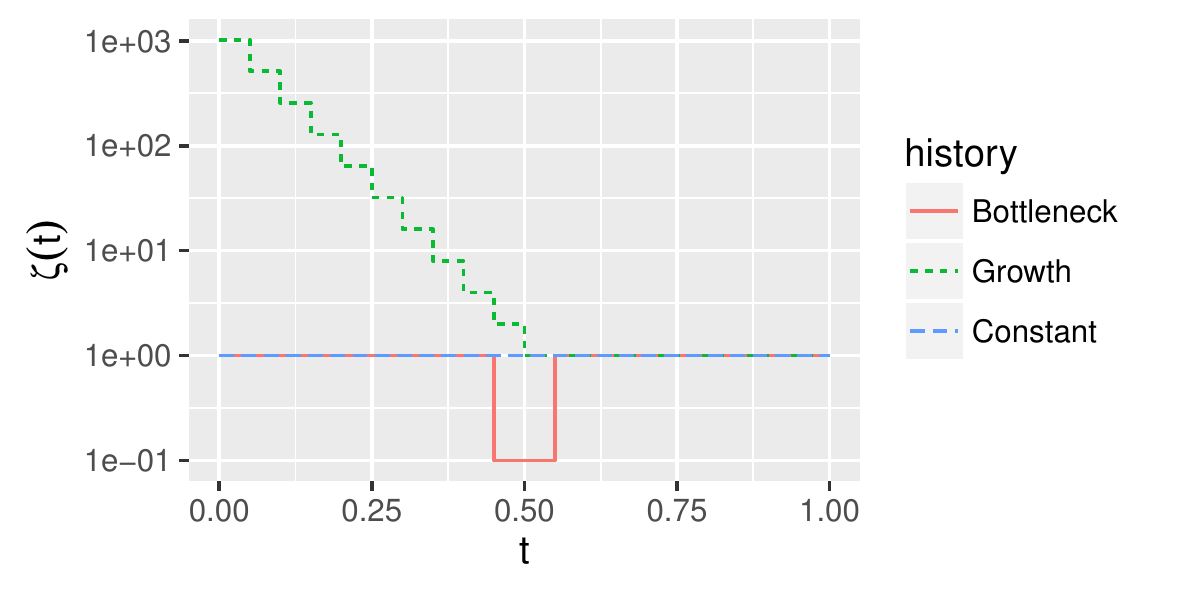} 
\caption{$\zeta(t)$ for three demographic scenarios: a constant size history, a bottleneck history that undergoes a temporary 10-fold size reduction, and a growth history with repeated population doublings. Note that the $y$-axis is stretched by $y \mapsto \log(y)$. }
\label{fig:histories}
\end{figure}

Using simulations, we now investigate the variation in the empirical SFS
across independent realizations of the coalescent process
and study the ability to infer the underlying model using the empirical SFS.
We consider three different $\zeta$, illustrated in Figure~\ref{fig:histories}.  Due to the association with population sizes in the case of Kingman's coalescent, we refer to $\zeta$ as the history or population size history.  However, we caution that depending on the finite population size model, $\zeta$ may not represent the population size, but some other biologically relevant parameter.   We consider a constant size history, a bottleneck history that undergoes a temporary 10-fold size reduction, and a growth history with repeated population doublings.
For each $\zeta$, we consider Beta${(2-\alpha,\alpha)}$-coalescents 
with $\alpha \in \{1,1.5,2\}$.
Note that $\alpha=1$ corresponds to the Bolthausen-Sznitman coalescent,
while $\alpha=2$ corresponds to the Kingman coalescent.
For each of the nine distinct values of $(\zeta, \alpha)$,
we simulated $m=1000$ independent trees with $n=20$ leaves.

In Figure~\ref{fig:branch:sqrt}, we examine the observed variation in branch
lengths across independent realizations of the coalescent process, from which we can deduce
the variation in the observed SFS.
Specifically, assume that each tree sampled from the coalescent process has the same mutation rate,
and, without loss of generality, assume that time has been scaled such that the mutation rate is 1.
Let $\branchLengths_{n,k}$ be
the sum of branch lengths with $k$ leaves and recall that $\hat{\tau}_{n,k}$ is the $k^{\text{th}}$ entry of the empirical SFS on the $n$ observed individuals.  Then, $\mathbb{P}(\hat{\bfmath{\tau}}_n \mid \bfmath{\branchLengths}_n) \sim
 \text{Poisson}(\bfmath{\branchLengths}_n)$,
and $\mathbb{E}[\bfmath{\branchLengths}_n] = \bfmath{\tau}_n$.
In Figure~\ref{fig:branch:sqrt}, we plot
$\branchLengths_{n,k}$ for each simulated tree,
as well as its expected value $\mu_{n,k}$.
Defining $\sigma_{n,k}^2 := \text{Var}(\branchLengths_{n,k})$ for this case of $m=1$, we also
plot an estimate of the standard deviation
$\hat{\sigma}_{n,k} = \sqrt{\hat{\mathbb{E}}[\branchLengths_{n,k}^2] - \tau_{n,k}^2}$, where $\hat{\mathbb{E}}$ is the empirical expectation.
Now, if we sum the branch lengths and mutations over $m$ independent trees (so then $\mathbb{E}[\branchLengths_{n,k}] = m \mu_{n,k}$, and $\text{Var}(\branchLengths_{n,k}) = m \sigma^2_{n,k}$), then
$\mu_{n,k}$ and  $\sigma_{n,k}^2$ describe the limiting
behavior of both $\branchLengths_{n,k}$ and $\hat{\tau}_{n,k}$ as $m \to \infty$:
by the Central Limit Theorem,
$\frac1{\sqrt{m}}(\branchLengths_{n,k} - \tau_{n,k}) \to_d \mathcal{N}(0, \sigma_{n,k}^2)$
and $\frac1{\sqrt{m}}(\hat{\tau}_{n,k} - \tau_{n,k}) \to_d \mathcal{N}(0, \sigma_{n,k}^2 + \mu_{n,k})$.

\begin{figure}[t]
  \centering
\includegraphics[width=\textwidth]{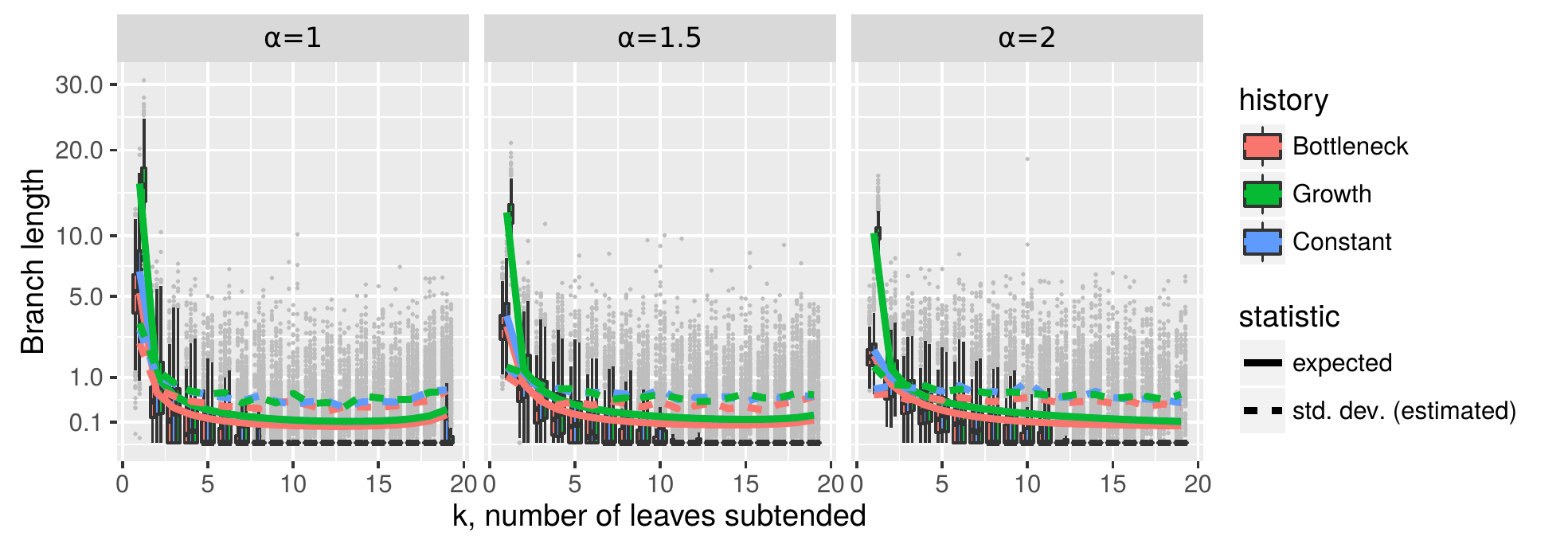}  
  \caption{The distribution of the branch length subtending $k$ leaves, for random trees under a Beta$(2-\alpha,\alpha)$-coalescent and $n=20$. The solid line is the expected value from Theorem~\ref{thm:main}.  We simulated 1000 independent trees per scenario and their branch length results are shown here as gray dots and box plots; the dashed line denotes the estimated standard deviation of the distribution. Note the $y$-axis is stretched by $y \mapsto \sqrt{y}$. The mean and standard deviation give the limiting behavior of $\hat{\tau}_{n,k}$ for many independent trees, under the Central Limit Theorem.  For most $k$ (say, $k \geq 5$), the branch length is usually $0$, and has high variance relative to the mean. Thus $\hat{\tau}_{n,k}$ will tend to have higher relative accuracy for the smaller entries $k$.}
\label{fig:branch:sqrt}
\end{figure}

\begin{figure}[t]
  \centering
\includegraphics[width=.88\textwidth]{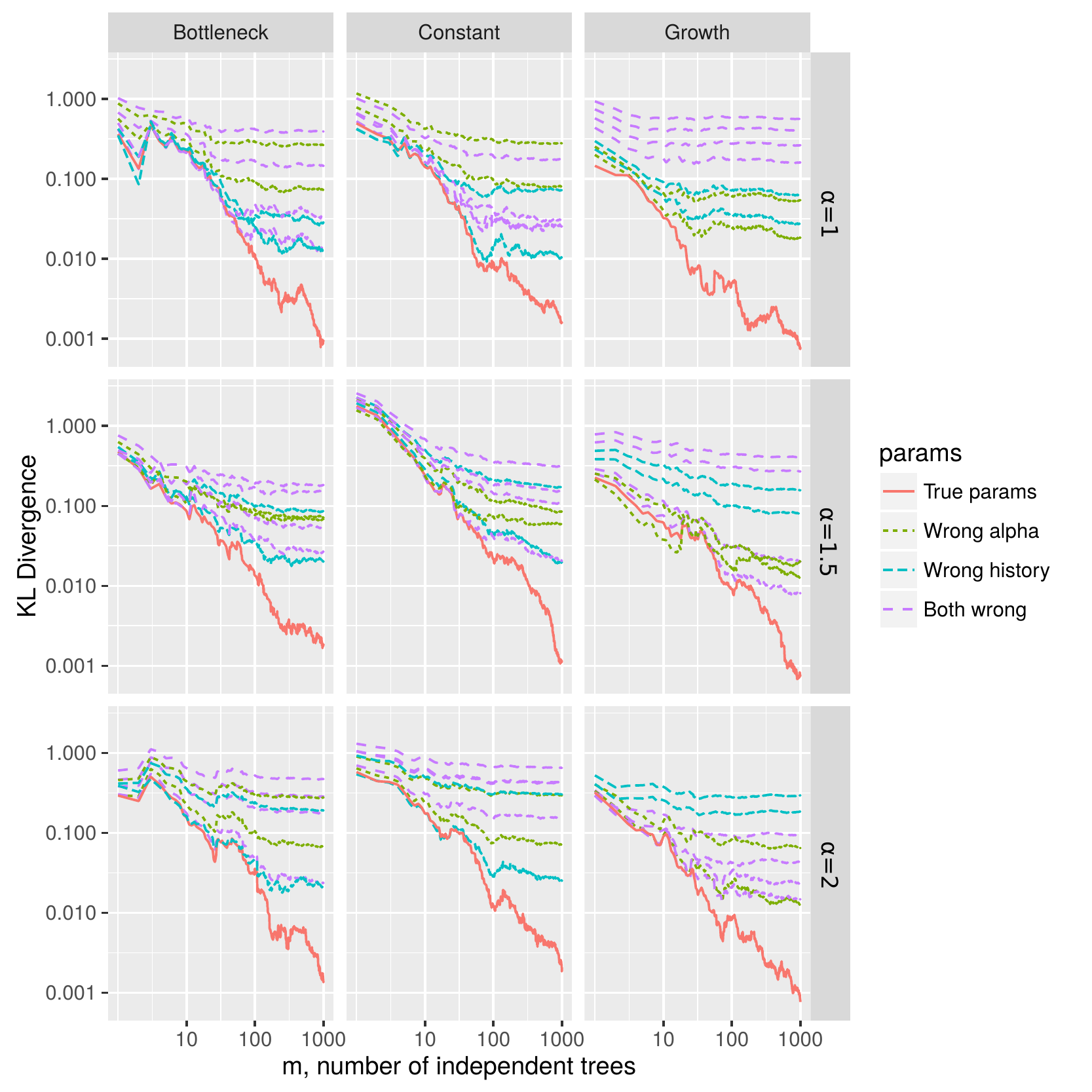}  
  \caption{The KL-Divergence $D_{KL}\big(\widetilde{P}^{(\zeta_1, \alpha_1)}_m \big\| P^{(\zeta_2,\alpha_2)} \big)$, where $P^{(\zeta_2,\alpha_2)}(k) \propto \tau_{n,k}^{(\zeta_2,\alpha_2)}$ is the distribution of derived alleles under scenario $(\zeta_2(t),\alpha_2)$, and $\widetilde{P}_{m}^{(\zeta_1,\alpha_1)}(k) \propto \branchLengths_{n,k}^{(\zeta_1, \alpha_1)}(m)$ is the conditional distribution of derived alleles, given the first $m$ trees simulated under $(\zeta_1(t), \alpha_1)$ and a mutation hitting one of those trees. For $m$ large enough, $D_{KL}$ is minimized by the true parameters, i.e. $(\zeta_1(t), \alpha_1(t)) = (\zeta_2(t), \alpha_2(t))$. $D_{KL}$ can typically discriminate the true scenario for $m=100$ trees. For $m=10$ trees, $D_{KL}$ is often, but not always, minimized by the true scenario.}
\label{fig:kl:div}
\end{figure}

A recent inconsistency result \citep[Theorem 1]{koskela2015bayesian}
shows that a $\Lambda$-measure cannot be inferred from a single tree ($m=1$),
even as $n \to \infty$.
Indeed, we see in Figure~\ref{fig:branch:sqrt} that the branch lengths $\branchLengths_{n,k}$ of a single tree can deviate substantially from $\tau_{n,k}$. For most $k$ (say, $k \geq 5$), typically $\branchLengths_{n,k} = 0$ or $\branchLengths_{n,k} \gg \mu_{n,k}$ given a single tree.  That is, for a single tree, branches subtending more than a few leaves are either not observed, or are much larger than the expected branch length. 
However, smaller $k$ (especially the singletons, $k=1$) have smaller relative standard deviation $\frac{\sigma_{n,k}}{\mu_{n,k}}$,
and thus will tend to have lower relative error $\frac{\hat{\tau}_{n,k} - \tau_{n,k}}{\tau_{n,k}} \approx \mathcal{N}(0, \frac{\sigma_{n,k}^2 + \mu_{n,k}}{m \mu_{n,k}^2} )$
as $m$ increases.

In the case of Kingman's coalescent, $\zeta$ is inferred by minimizing the KL-divergence between a normalized version of the empirical SFS and a normalized version of the expected SFS (e.g. \citet[Equation~10]{bhaskar2015efficient}).  We investigate how KL-divergence behaves as a function of the number $m$ of independent trees simulated in the case of $\Lambda$-coalescents.  Let $\bfmath{\tau}_n^{(\zeta,\alpha)}$ be the expected SFS under model
$(\zeta(t), \alpha)$, and $\bfmath{\branchLengths}_n^{(\zeta,\alpha)}(m)$
the corresponding branch lengths summed over the first $m$ simulated trees.
Define $P^{(\zeta,\alpha)}(k) \propto {\tau}_{n,k}^{(\zeta,\alpha)}$ as the true probability
distribution of derived alleles under scenario $(\zeta(t), \alpha)$,
and $\widetilde{P}_{m}^{(\zeta,\alpha)}(k) \propto \branchLengths_{n,k}^{(\zeta, \alpha)}(m)$ as the conditional distribution of derived alleles, given the first $m$ trees simulated under $(\zeta(t), \alpha)$.
In Figure~\ref{fig:kl:div}, we plot
the KL-Divergence $D_{KL}\big(\widetilde{P}^{(\zeta_1, \alpha_1)}_m \big\| P^{(\zeta_2,\alpha_2)} \big)$
as a function of $m$, for every $(\zeta_1(t), \zeta_2(t), \alpha_1, \alpha_2)$ considered above (that is, $\zeta$ is constant, bottleneck, or growth, and $\alpha$ is $1$, $1.5$, or $2$).
In this case, we see that minimizing $D_{KL}$ identifies
the true scenario $(\zeta_1(t), \alpha_1(t)) = (\zeta_2(t), \alpha_2(t))$ with access to only a moderate number of independent trees (between 10 to 100).

Figure~\ref{fig:kl:div} is encouraging, as not too many independent
trees are needed to distinguish between the different scenarios $(\zeta(t),\alpha)$.
Unfortunately, in some cases it may be impossible to even sample
two independent trees (personal communication, Jere Koskela).
For example, in the model of \citet{birkner2013ancestral},
a multiple merger event happens over a single ``generation'',
which can cause the multiple merger to affect unlinked sites,
resulting in correlated coalescence times.
However, in other models, multiple merger events may only affect the genome locally, and thus trees from unlinked sites are independent.
For example, in the selective sweep model of \citet{durrett2005coalescent},
multiple mergers are caused by selective sweeps taking place over $O(\log(N))$ ``generations'', and a site experiences a multiple merger if $\frac{r_N \log(2N)}{s_N}=O(1)$,
where $N,s_N,r_N$ respectively parametrize
the population size, selection strength, and recombination distance to the selected site.
Thus, the independence of unlinked trees is not necessarily determined by
the $\Lambda$- or $\Xi$-measure itself, but instead by the pre-limiting model.

%%%%%%%%%%%%%%%%%%%%%%%%%%%%%%%%%%%
\section{Identifiability Results}\label{sec:identifiability}
%%%%%%%%%%%%%%%%%%%%%%%%%%%%%%%%%%%

Before attempting to infer $\zeta$ or $\Xi$ in practice, it is important to know whether such inference is possible using the SFS.  For instance, when inferring $\zeta$, if two different functions $\zeta_1$ and $\zeta_2$ produce the same SFS, then it is impossible to distinguish between the two using only the SFS.  In such a case, we say that $\zeta$ is not identifiable.  For Kingman's coalescent if one allows $\zeta$ to be an arbitrary positive function that produces a finite SFS, then $\zeta$ is not identifiable \citep{myers2008can}.  $\zeta$ is identifiable in the case of Kingman's coalescent, however, if one restricts $\zeta$ to be from a set of biologically realistic functions (technically, a set of functions with only a finite number of oscillations) \citep[Theorem 11]{bhaskar2014descartes}.  We show that a similar result holds for all coalescents of the form $\frac{\Xi(d\bfx)}{\zeta(t)}$ where $\Xi(d\bfx)$ is fixed.

In general it is impossible to infer $\Xi$ from the SFS if $\Xi$ is not restricted.  There has been some interest, however, in the case of distinguishing between a subset of $\Lambda$-coalescents \citep{eldon2015can}.  We prove some results about the identifiability of the measure for various subsets of $\Lambda$-measures when $\zeta$ is a constant function.  We also consider the question posed by \citet{eldon2015can} of whether or not the SFS can distinguish between exponential growth under Kingman's coalescent and a class of $\Lambda$-coalescents with constant $\zeta$, and we show that indeed it is possible to distinguish between these cases with a surprisingly small number of samples.  We note that our identifiability results require knowledge of the exact expected SFS, whereas \citet{eldon2015can} focus on the case where the expected SFS is approximated using an empirical SFS, which is what occurs in practice.

Throughout this section we assume that one has the exact expected SFS (i.e., the object computed by Theorem~\ref{thm:main}).  

%%%%%%%%%%%%%%%%%%%%%%%%%%%%%%%%%%
\subsection{Identifiability of $\zeta$ for fixed $\Xi$ measure}
%%%%%%%%%%%%%%%%%%%%%%%%%%%%%%%%%%
Before proceeding to the results and proofs, we first introduce some notation.  Let $\mathcal{M}_K(\mathcal{F})$ denote the set of piecewise defined functions with at most $K$ pieces made from some function family $\mathcal{F}$.  Furthermore, let $\sS(\mathcal{F})$ denote the sign-change complexity of $\mathcal{F}$.  Informally, $\sS(\mathcal{F})$ is the supremum of the number of times $f_1 - f_2$ crosses 0 over functions $f_1,f_2\in\mathcal{F}$, which is related to the number of oscillations each $f \in \mathcal{F}$ is allowed to have (see \citet[Definition~4]{bhaskar2014descartes} for a formal definition of $\sS(\mathcal{F})$).  We will also write $\psi^{\Xi}_n$ for the number of 0 entries in $[\bold{U}^{-1}]_{\cdot,1}$ in the spectral decomposition of $\bold{Q}$ for a coalescent on $n$ individuals governed by $\Xi(d\bfx)$.  Furthermore, denote by $\mathcal{X}$ the space of $\Xi$ measures such that $(\bold{Q})_{k,k-1} > 0$ for all $k$.  That is, $\mathcal{X}$ is the set of $\Xi$ measures where for any sample size there is positive probability of a single pairwise merger.  If we are only considering $\Lambda$-coalescents, then $\mathcal{X}$ contains all $\Lambda$ measures except for $\delta_1$, the star coalescent.  We now present our main identifiability results and a conjectured bound on $\psi^{\Xi}_n$.

Our main result on the identifiability of $\zeta$ is the following theorem.
\begin{theorem}
\label{thm:eta}
For an arbitrary $\Xi$-coalescent governed by the measure $\frac{\Xi(d\bfx)}{\zeta(t)}$ where $\Xi \in \mathcal{X}$ is fixed, suppose $\sS(\mathcal{F}) < \infty$ and $n \ge 2K + (2K-1)\sS(\mathcal{F}) + \psi^\Xi_n$.  Then for each expected SFS $\bftau_n$ there exists a unique $\zeta \in \mathcal{M}_K(\mathcal{F})$ consistent with $\bftau_n$.
\end{theorem}
First, note that in the case of Kingman's coalescent, $\psi^{\delta_0}_n = 0$ for all $n$, and so in some sense, Kingman's coalescent is optimal in terms of the number of samples needed to ensure that a certain model space is identifiable.  For the Bolthausen-Sznitman coalescent, $\psi^1_n = 0$ for all $n$, which follows from the spectral decomposition \citep[Theorem~1.1]{moehle2014spectral}.  For the point mass $\Lambda$-coalescent with mass at $\frac{1}{2}$, for $n \ge 5$, all odd entries of $[\bold{U}^{-1}]_{\cdot,1}$ are 0 and so $\psi^{\delta_{1/2}} > 0$, thus implying that larger samples (relative to Kingman's coalescent or the Bolthausen-Sznitman coalescent) are needed for this coalescent to ensure that a given model space is identifiable.  We suspect that $\Lambda(dx) = \delta_{1/2}$ is the worst case among all $\Xi$-coalescents in $\mathcal{X}$ for identifiability, resulting in the following conjecture:

\begin{conjecture}  
	For all $\Xi \in \mathcal{X}$ and $n \ge 3$, $\psi^\Xi_n \le \lfloor \frac{n-1}{2} \rfloor -1$.
\end{conjecture}
If this conjecture is true, then the bound on the sample size needed to have identifiability in Theorem~\ref{thm:eta} can be simplified to $n \ge 2[2K + (2K-1)\sS(\mathcal{F})]$.  

%%%%%%%%%%%%%%%%%%%%%%%%%%%%%%%%%%
\subsection{Identifiability of the $\Lambda$ measure for a constant $\zeta$}
%%%%%%%%%%%%%%%%%%%%%%%%%%%%%%%%%%
We also have the following results for $\Lambda$-coalescents about the identifiability of the $\Lambda$ measure.
\begin{theorem}  
\label{thm:point}
Consider the set of point-mass $\Lambda$-coalescents: $\left\{\delta_z : z \in [0,1] \right\}$.  If $\Lambda$ is restricted to be in this set and $n \ge 3$, then the expected SFS $\bftau_n$ uniquely determines $\Lambda$.
\end{theorem}
\begin{theorem} 
\label{thm:beta}
Consider the set of $\textBeta$-coalescents: $\left\{\mathcal{L}(X) : X \sim \textBeta(2-\alpha,\alpha), \alpha \in [1,2) \right\}$.  If $\Lambda$ is restricted to be in this set and $n \ge 3$, then the expected SFS $\bftau_n$ uniquely determines $\Lambda$.
\end{theorem}
\begin{theorem} 
\label{thm:eldon}
Consider the set of coalescents: $\left\{\frac{\delta_0}{ae^{-bt}} : a,b>0 \right\} \cup \left\{ \delta_z : z \in [0,1] \right\} \cup\break \left\{\mathcal{L}(X) : X \sim \textBeta(2-\alpha,\alpha), \alpha \in [1,2) \right\}$, that is: Kingman's coalescent with exponential growth, point-mass coalescents, or $\textBeta$-coalescents.  If $\Lambda$ is restricted to be in this set and $n \ge 4$, then the expected SFS $\bftau_n$ uniquely determines $\Lambda$.
\end{theorem}
Theorem~\ref{thm:eldon} gives a positive theoretical answer to the question of whether or not the SFS can distinguish between exponential growth and multiple-merger coalescents. Using the techniques presented below, it is straightforward to obtain similar results for other subsets of $\Lambda$-coalescents.

%%%%%%%%%%%%%%%%%%%%%%%%%%%%%%%%%%
\subsection{Proofs of the identifiability results}
%%%%%%%%%%%%%%%%%%%%%%%%%%%%%%%%%%
The following Lemma will be used in proving the theorems in this section and may be of independent interest, as it shows that given the SFS for $n$ individuals one can compute the expected time to most recent common ancestor for sample sizes $2,\ldots,n$ or vice-versa.

\begin{lemma}  
\label{lem:bijection}
For all $\Lambda$- and $\Xi$-coalescents, there is a bijection between the expected SFS $\bftau_n$ and the expected times $\bfa_n$ to most recent common ancestor.
\end{lemma}
\begin{proof}
Combine Lemmas~\ref{lem:anti} and \ref{lem:tmrca} to see that $\bftau_n = \bold{B} \bold{C} \bfa_n$, with $\bold{B}$ and $\bold{C}$ being universal.  Then, since $\bold{B}$ is upper triangular and all of its diagonal entries are non-zero it is invertible.  Furthermore, since $\bold{C}$ is bi-diagonal and the diagonal entries are all non-zero, it is also invertible.  Therefore, $\bold{B} \bold{C}$ is invertible and since $\bftau_n$ and $\bfa_n$ are related through an invertible matrix the transformation is bijective.
\end{proof}
To prove Theorem~\ref{thm:eta} we will use the following lemma.

\begin{lemma}
\label{lem:monotone} 
Let $\lambda_k = -(\bold{Q}_{kk})$.
For all $\Xi \in \mathcal{X}$ and all $\Lambda$ other than $\Lambda(dx) = \delta_1(dx)$ (i.e., the star coalescent), the sequence $(\lambda_k)_{k \ge 2}$ is strictly increasing.
\end{lemma}
\begin{proof}
Consider a sample of size $k+1$ and a subsample of size $k$.  Without loss of generality, assume you remove individual $k+1$ to produce the subsample.  The time to the first event is the same for both samples unless the first event only involves individual $k+1$ and one lineage from $\left\{1,\ldots,k\right\}$.  That is, the total rate when there are $k+1$ lineages is equal to the total rate when there are $k$ lineages plus $k$ times the rate at which exactly a particular pair of individuals coalesce. Formally,
\[
\lambda_{k+1} = \lambda_{k} + \frac{k}{{k+1 \choose 2}} (\bold{Q})_{k+1, k}.
\]
By assumption, $(\bold{Q})_{k+1,k} > 0$, and so the total rates must be strictly increasing.
\end{proof}

We now prove Theorem~\ref{thm:eta}.  Our proof relies heavily on the proof of the corresponding result for Kingman's coalescent \citep[Theorem~11]{bhaskar2014descartes}.  We essentially show that this setting satisfies the same hypotheses as the Kingman's coalescent case and then use that result to complete our proof. 

\begin{proof}[Proof of Theorem~\ref{thm:eta}]  By Lemma~\ref{lem:bijection}, the SFS is uniquely determined by $\bfa_n$.  Then, furthermore, note that from Lemma~\ref{lem:spectral} the matrix $\bold{U}$ is invertible since it is triangular with all non-zero entries along the diagonal.  Then, by the same argument as in \citet[Equation~12]{bhaskar2014descartes}, we know that if the model space is not identifiable then for each $k$ not corresponding to a zero in $\bold{D}$ (contributing to $\psi_n^{\Xi}$), $\lambda_k$ must be the root of the Laplace transform of two different functions in the model space.  By Lemma~\ref{lem:monotone}, these are all distinct, resulting in $n - \psi_n^{\Xi}$ roots.  Then, by taking $n - \psi_n^{\Xi}$ sufficiently large, we obtain a contradiction via the Generalized version of Descarte's Rule of Signs \citep[Theorem~4]{bhaskar2014descartes} and the theorem is proved.
\end{proof}
We now prove Theorems~\ref{thm:point}, \ref{thm:beta}, and \ref{thm:eldon}.  The idea is to explicitly calculate the $\mathbb{E}\tmrca{k}$ for the first few $k$ for each allowed $\Lambda$ measure and then use Lemma~\ref{lem:bijection} to show that if $\Lambda$ is uniquely determined by the first few $\mathbb{E}\tmrca{k}$, then it is uniquely determined by $\bftau_n$.

\begin{proof}[Proof of Theorem~\ref{thm:point}]  $\mathbb{E}\tmrca{2} = 1$ for all $\Lambda$ in the set of possible $\Lambda$s. Consider  $\Lambda = \delta_z$.  Using Lemma~\ref{lem:anconstant} we see that:
\begin{align}
\mathbb{E}\tmrca{3} &= \frac{1}{3-2z} + \frac{3-3z}{3-2z} \cdot 1 =  \frac{4-3z}{3-2z}. \label{eq:ET3_point}
\end{align}
This is a monotonically decreasing function of $z \in [0,1]$, and so $\Lambda$ is uniquely determined by $\mathbb{E}\tmrca{3}$.  Then, appealing to Lemma~\ref{lem:bijection}, we see that $\Lambda$ is uniquely determined by the SFS for $n \ge 3$.
\end{proof}

\begin{proof}[Proof of Theorem~\ref{thm:beta}] A calculation similar to \eqref{eq:ET3_point} gives $\mathbb{E}\tmrca{3} = \frac{2+3\alpha}{2+2\alpha}$ for $\Lambda = \mathcal{L}(X), X \sim \textBeta(2-\alpha,\alpha)$, where $\alpha \in [1,2)$.  This is a monotonically increasing function of $\alpha \in [1,2)$ and the claim follows from the same argument as in the proof of Theorem~\ref{thm:point}.
\end{proof}

\begin{proof}[Proof of Theorem~\ref{thm:eldon}] Suppose that two distinct $\Lambda$-coalescents within the set of allowed models produce the same expected SFS for $n\geq 4$.  Then, by Lemma~\ref{lem:bijection}, they would have the same values of $\mathbb{E}\tmrca{2}$, $\mathbb{E}\tmrca{3}$, and $\mathbb{E}\tmrca{4}$.  By Theorems~\ref{thm:point} and \ref{thm:beta}, we know that the $\Lambda$ measures cannot both be point-mass coalescents or $\textBeta$-coalescents.  From \citet[Corollary~8]{bhaskar2014descartes}, we also know that the $\Lambda$ measures cannot both be Kingman's coalescent with different exponential growth parameters.  There are thus three cases.  They are all straightforward, albeit tedious.

Case 1: one $\Lambda$ measure is a point-mass coalescent and the other is a $\textBeta$-coalescent.  Letting  $\mathbb{E}\tmrca{2} = 1$ (without loss of generality), we can explicitly compute $\mathbb{E}\tmrca{4}$ for the point-mass coalescent and the $\textBeta$-coalescent using the same recursive idea as in the proof of Theorem~\ref{thm:point}.  Let $p_{m,k}$ denote the probability that when there are $m$ lineages exactly $k$ of them are involved in the next coalescence event.  Then, by Lemma~\ref{lem:anconstant}
$\mathbb{E}\tmrca{4} = c_{4,4} + p_{4,2} \mathbb{E}\tmrca{3} + p_{4,3} \mathbb{E}\tmrca{2}$.
In particular, for the point-mass coalescent $\delta_z$, this implies 
$\mathbb{E}\tmrca{4}  
= \frac{5}{3} + \frac{1}{2z-3} + \frac{3-2z}{18-24z+9z^2}.$
Now, recalling the expression of $\bbE\tmrca{3}$ in \eqref{eq:ET3_point} and letting
\begin{equation}
	\mathbb{E}\tmrca{3} = t
\label{eq:ET3_t}	
\end{equation}
implies $z=\frac{3t-4}{2t-3}$.  Plugging this into $\mathbb{E}\tmrca{4}$, we see that for the point-mass coalescent,
\begin{equation}
\mathbb{E}\tmrca{4} = \frac{1}{3}\bigg[6t-4 - \frac{2t-3}{6+t(3t-8)}\bigg].
\label{eq:ET4_point}
\end{equation}
A similar calculation for the $\textBeta$-coalescent shows that
\begin{equation}
\mathbb{E}\tmrca{4} = \frac{1}{3}\left(6t -2 + \frac{1}{t-2} \right),
\label{eq:ET4_beta}
\end{equation}
with $t :=\mathbb{E}\tmrca{3}$ under the $\textBeta$-coalescent.
Equating \eqref{eq:ET4_point} and \eqref{eq:ET4_beta}, and solving for $t$ results in the solution $t=1$ or $t = \frac{4}{3}$.  But, if $t=1$, then we see that $z=1$, the star-coalescent, which corresponds to $\alpha = 0$ for the $\textBeta$-coalescent, which is not in the set of allowed $\textBeta$-coalescents.  If $t=\frac{4}{3}$, we see that $z=0$, which corresponds to Kingman's coalescent, and $\alpha=2$ for the $\textBeta$-coalescent, which again, is not in the set of allowed $\textBeta$-coalescents.  Therefore, a point-mass coalescent and a $\textBeta$-coalescent with $\alpha\in[1,2)$ cannot have the same $\mathbb{E}\tmrca{2},\mathbb{E}\tmrca{3}$ and $\mathbb{E}\tmrca{4}$ simultaneously.  

Case 2: one $\Lambda$ measure is a point-mass coalescent and the other is Kingman's coalescent with exponential growth.  
Without loss of generality, assume that  $\mathbb{E}\tmrca{2} = 1$ for the point-mass $\Lambda$-coalescent.
The exponential-growth Kingman's coalescent model considered here has
$c_{m,m} = -\frac{1}{b} e^{{m \choose 2}/(ab)} \Ei[- {m \choose 2}/(ab)]$,  
where $\Ei(x) := -\int_{-x}^{\infty} \frac{e^{-t}}{t}dt$ is the exponential integral
\citep[Supplemental Material Equation 5]{bhaskar2015efficient}.
Then, the constraint $\mathbb{E}\tmrca{2} = c_{2,2}=1$ implies $b = - e^{1/d} \Ei(- 1/d)$, where $d:= ab$.
Furthermore, assuming this constraint and applying Theorem~\ref{thm:main} to Kingman's coalescent, we obtain 
\begin{align}
\mathbb{E}\tmrca{3} &= \frac{3}{2} - \frac{c_{3,3}}{2}= \frac{3}{2} - \frac{e^{2/d}\Ei(-3/d)}{2\Ei(-1/d)},
\label{eq:ET3_kingman}\\
\mathbb{E}\tmrca{4} &=  \frac{9}{5} - c_{3,3} + \frac{c_{4,4}}{5} = \frac{9}{5} - \frac{e^{2/d}\Ei(-3/d)}{\Ei(-1/d)} + \frac{e^{5/d}\Ei(-6/d)}{5\Ei(-1/d)}.
\label{eq:ET4_kingman}
\end{align}
Now, in addition to $\mathbb{E}\tmrca{2}$, if the two coalescents have the same values of $\mathbb{E}\tmrca{3}$ and $\mathbb{E}\tmrca{4}$, then the right hand sides of \eqref{eq:ET3_t} and \eqref{eq:ET3_kingman} must agree, while the right hand sides of \eqref{eq:ET4_point} and \eqref{eq:ET4_kingman} must agree.  This implies
\begin{align}
\frac{ f_1(d) + e^{\frac{5}{d}}\Ei(-\frac{6}{d}) f_2(d)}{\Ei(-\frac{1}{d})f_2(d)}=0,
\label{eq:f_condition}
\end{align}
where
\begin{align*}
f_1(d) &:= 2\Eip{-\frac{1}{d}}\bigg\{e^{\frac{4}{d}} \bigg[\Eip{-\frac{3}{d}}\bigg]^2 - 4 e^{\frac{2}{d}}\Eip{-\frac{3}{d}}\Eip{-\frac{1}{d}} + \bigg[\Eip{-\frac{1}{d}}\bigg]^2 \bigg\},\\
f_2(d) &:=  3e^{\frac{4}{d}} \bigg[\Eip{-\frac{3}{d}}\bigg]^2 - 2e^{\frac{2}{d}}\Eip{-\frac{3}{d}}\Eip{-\frac{1}{d}} + 3\bigg[\Eip{-\frac{1}{d}}\bigg]^2.
\end{align*}
However, by Lemma~\ref{lem:f} in the appendix, there is no $d\in(0,\infty)$ such that \eqref{eq:f_condition} holds.

Case 3: one $\Lambda$ measure is a $\textBeta$-coalescent and the other is Kingman's coalescent with exponential growth.  
 If these two coalescents produce the same values of $\mathbb{E}\tmrca{2}, \mathbb{E}\tmrca{3}$ and $\mathbb{E}\tmrca{4}$, then we must have $t=\frac{3}{2} - \frac{e^{2/d}\Ei(-3/d)}{2\Ei(-1/d)}$ in \eqref{eq:ET4_beta}, and equating \eqref{eq:ET4_beta} and \eqref{eq:ET4_kingman} implies
\begin{align}
\frac{ g_1(d) + 3e^{\frac{5}{d}}\Ei(-\frac{6}{d}) g_2(d)}{\Ei(-\frac{1}{d})g_2(d)}=0,
\label{eq:g_condition}
\end{align}
where
\begin{align*}
g_1(d) &:= 2\Eip{-\frac{1}{d}} \bigg[ -4 e^{\frac{2}{d}} \Eip{-\frac{3}{d}} + \Eip{-\frac{1}{d}} \bigg],\\
g_2(d) &:= e^{\frac{2}{d}} \Eip{-\frac{3}{d}} + \Eip{-\frac{1}{d}}.
\end{align*}
However, by Lemma~\ref{lem:g} in the appendix, there is no $d\in(0,\infty)$ such that \eqref{eq:g_condition} holds.

Since each of the three cases results in a contradiction, we see that no such $\Lambda$ measures exist, proving the identifiability claim.
\end{proof}

%%%%%%%%%%%%%%%%%%%%%%%%%%%%%%%%%%%
\section{Discussion}\label{sec:discussion}
We have presented an efficient algorithm for computing the SFS for a very general class of coalescents.  While $\Lambda$- and $\Xi$-coalescents seem to be primarily used in practice to model the genealogies of marine species \citep{arnason2004mitochondrial, hedgecock2011sweepstakes}, these coalescents also model a wide range of other phenomena including continuous strong positive selection \citep{neher2013genealogies}, recurrent selective sweeps \citep{durrett2004approximating,durrett2005coalescent}, strong bottlenecks \citep{birkner2009modified} and many others.  Perhaps one of the reasons these coalescents are less widely used than Kingman's coalescent is because efficient inference tools have not yet been developed to the same extent.

Multiple-merger coalescents have also attracted some interest recently in the context of extremely large sample sizes \citep{bhaskar2014distortion}.  In such cases the sample size is too large for the assumption of only pairwise mergers of lineages imposed by Kingman's coalescent to be biologically plausible, and indeed using Kingman's coalescent to model such populations causes biases in inference \citep{bhaskar2014distortion}.  It should be possible to extend the results presented in this paper to discrete-time coalescents, such as the ``exact coalescent'' \citep{fu2006exact} corresponding to the coalescent arising from the discrete-time Wright-Fisher process, or any of the discrete-time random mating models considered by \citet{eldon2006coalescent}.

We also presented some encouraging identifiability results.  While it is impossible in the general case to infer the inverse intensity function $\zeta$ or the measure of a $\Lambda$-coalescent from the SFS, for many biologically important cases identifiability does indeed hold.  The method we presented for proving that the $\Lambda$ measure is identifiable for constant $\zeta$ is powerful, but straightforward and should make it easy to prove whether or not the measure is identifiable for other sets of $\Lambda$- or $\Xi$-coaleascents.  While we only considered the identifiability of $\Lambda$ for fixed, constant $\zeta$ and the identifiability of $\zeta$ for fixed $\Lambda$ or $\Xi$, it would be interesting to see if identifiability results can still be obtained for some model spaces while allowing both $\Lambda$ and $\zeta$ to vary.  It would also be interesting to extend our identifiability results for the $\Lambda$ measure to some of the biologically relevant $\Xi$-coalescents.

Our identifiability results generally assumed access to the expected SFS.  In practice, one observes a finite number of sites and so one only has a noisy estimate of the SFS.  
Our simulation study shows that, given a moderate number of independent trees, the empirical SFS is accurate enough
to distinguish $\frac{\Lambda(dx)}{\eta(t)}$ for some simple models.
However, the effect of noisy data is still largely unknown,
especially in cases where convergence to the expected SFS is not guaranteed.
The accuracy of inferring $\zeta$ with the empirical SFS has been studied for Kingman's coalescent \citep{terhorst2015fundamental}, 
and it would be interesting to extend these results to general $\Lambda$-coalescents, and to the inference of the $\Lambda$-measure itself;
the results presented here should make such an analysis more tractable.

\section{Acknowledgments}
We thank Jere Koskela for helpful discussion on convergence to the expected SFS.  This research is supported in part by an NIH grant R01-GM108805, an NIH training grant T32-HG000047, and a Packard Fellowship for Science and Engineering.

\appendix

\section{Appendix}

Here we present two lemmas that are used in Theorem~\ref{thm:eldon}.  Proofs are tedious but straightforward.
\begin{lemma}
\label{lem:f}
For $d\in(0,\infty)$,
\begin{align*}
\frac{ f_1(d) + e^{\frac{5}{d}}\Ei(-\frac{6}{d}) f_2(d)}{\Ei(-\frac{1}{d})f_2(d)}\ne0,
\end{align*}
where
\begin{align*}
f_1(d) &:= 2\Eip{-\frac{1}{d}}\bigg\{e^{\frac{4}{d}} \bigg[\Eip{-\frac{3}{d}}\bigg]^2 - 4 e^{\frac{2}{d}}\Eip{-\frac{3}{d}}\Eip{-\frac{1}{d}} + \bigg[\Eip{-\frac{1}{d}}\bigg]^2 \bigg\},\\
f_2(d) &:=  3e^{\frac{4}{d}} \bigg[\Eip{-\frac{3}{d}}\bigg]^2 - 2e^{\frac{2}{d}}\Eip{-\frac{3}{d}}\Eip{-\frac{1}{d}} + 3\bigg[\Eip{-\frac{1}{d}}\bigg]^2.
\end{align*}
\end{lemma}
\begin{lemma}
\label{lem:g}
For $d\in(0,\infty)$,
\begin{align*}
\frac{ g_1(d) + 3e^{\frac{5}{d}}\Ei(-\frac{6}{d}) g_2(d)}{\Ei(-\frac{1}{d})g_2(d)}\ne0,
\end{align*}
where
\begin{align*}
g_1(d) &:= 2\Eip{-\frac{1}{d}} \bigg[ -4 e^{\frac{2}{d}} \Eip{-\frac{3}{d}} + \Eip{-\frac{1}{d}} \bigg],\\
g_2(d) &:= e^{\frac{2}{d}} \Eip{-\frac{3}{d}} + \Eip{-\frac{1}{d}}.
\end{align*}
\end{lemma}

In what follows, let $\Eone(x) := \int_{x}^{\infty} \frac{e^{-t}}{t} dt = - \text{Ei}(-x)$.  It is clear that $\Eone(x) > 0$ for all $x > 0$.  Additionally,
\begin{align}
e^{\frac{n}{d}}\Eone\left(\frac{n+1}{d}\right) = \int_{\frac{1}{d}}^{\infty} \frac{e^{-t}}{t + \frac{n}{d}} dt \label{eq:e1identity},
\end{align}
which follows from the definition of $\Eone$ and a change of variables.
%%%%Proofs%%%%
\begin{proof}[Proof of Lemma~\ref{lem:f}]
First, by noting that $f_2(d) = 3[e^\frac{2}{d}\Eone(\frac{3}{d}) - \Eone(\frac{1}{d})]^2 + 4e^\frac{2}{d}\Eone(\frac{3}{d})\Eone(\frac{1}{d})$, it is easy to see that the denominator is strictly negative for $d \in (0,\infty)$.  We will now show that the numerator is strictly positive for $d \in (0,\infty)$.  First, by rearranging terms we see that
\begin{align}
\label{eq:f_first_step}
f_1(d) + e^{\frac{5}{d}}\Eip{-\frac{6}{d}} f_2(d) &= 4 \bigg[ \Eone\left(\frac{1}{d}\right) - e^{\frac{5}{d}} \Eone\left(\frac{6}{d}\right)\bigg]e^\frac{2}{d}\Eone\left(\frac{3}{d}\right) \Eone\left(\frac{1}{d}\right)\\
\nonumber &\hspace{.5in}-\bigg[ \Eone\left(\frac{1}{d}\right) - e^{\frac{2}{d}} \Eone\left(\frac{3}{d}\right)\bigg]^2\bigg[2\Eone\left(\frac{1}{d}\right) + 3 e^\frac{5}{d}\Eone\left(\frac{6}{d}\right)\bigg].
\end{align}

Then, note
\begin{align*}
\Eone\left(\frac{1}{d}\right) - e^{\frac{2}{d}} \Eone\left(\frac{3}{d}\right) &= \int_\frac{1}{d}^\infty \frac{e^{-t}}{t(t+\frac{2}{d})} dt < \frac{4}{d} \int_\frac{1}{d}^\infty \frac{e^{-t}}{t(t+\frac{5}{d})} = \frac{4}{5}\bigg[\Eone\left(\frac{1}{d}\right) - e^{\frac{5}{d}} \Eone\left(\frac{6}{d}\right)\bigg].
\end{align*}
Applying this inequality to the negative term on the right hand side of \eqref{eq:f_first_step}, we see
\begin{align*}
\lefteqn{f_1(d) + e^{\frac{5}{d}}\Eip{-\frac{6}{d}} f_2(d)} & \\
& > 4 \bigg[ \Eone\left(\frac{1}{d}\right) - e^{\frac{5}{d}} \Eone\left(\frac{6}{d}\right)\bigg] \\
& \hspace{5mm}\times \bigg\{ e^\frac{2}{d}\Eone\left(\frac{3}{d}\right)\Eone\left(\frac{1}{d}\right)
- \bigg[\frac{4}{5d}\Eone\left(\frac{1}{d}\right) + \frac{6}{5d}e^\frac{5}{d}\Eone\left(\frac{6}{d}\right)\bigg] \bigg(\int_\frac{1}{d}^\infty \frac{e^{-t}}{t(t+\frac{2}{d})} dt\bigg) \bigg\}\\
&>  4 \bigg[ \Eone\left(\frac{1}{d}\right) - e^{\frac{5}{d}} \Eone\left(\frac{6}{d}\right)\bigg] \Eone\left(\frac{1}{d}\right) \bigg[ e^{\frac{2}{d}}\Eone\left(\frac{3}{d}\right) - \frac{4}{15}\Eone\left(\frac{1}{d}\right) - \frac{2}{5}e^\frac{5}{d}\Eone\left(\frac{6}{d}\right)\bigg]\\
&= 4 \bigg[ \Eone\left(\frac{1}{d}\right) - e^{\frac{5}{d}} \Eone\left(\frac{6}{d}\right)\bigg] \Eone\left(\frac{1}{d}\right)\bigg[ \int_\frac{1}{d}^\infty \frac{(\frac{1}{3}t^2+\frac{7}{3d}t - \frac{8}{3d^2}) e^{-t}}{t(t+\frac{2}{d})(t+\frac{5}{d})}\bigg],
\end{align*}
which is greater than 0 for any $d \in (0,\infty)$ since $\Eone\left(\frac{1}{d}\right) > e^{\frac{5}{d}} \Eone\left(\frac{6}{d}\right)$ and $\frac{1}{3}t^2+\frac{7}{3d}t - \frac{8}{3d^2} > 0$ for $t > \frac{1}{d}$.
\end{proof}
%%%%%%%%%%%
\begin{proof}[Proof of Lemma~\ref{lem:g}]
The denominator of \eqref{eq:g_condition} is equal to $\Eone\left(\frac{1}{d}\right)\big[ e^{\frac{2}{d}}\Eone\left(\frac{3}{d}\right) + \Eone\left(\frac{1}{d}\right)\big]$ which is strictly positive for $d \in (0,\infty)$, by definition of $\Eone(x)$.  Furthermore, the numerator is strictly negative for $d \in (0,\infty)$ by noting the following:
\begin{align*}
\lefteqn{g_1(d) + 3e^{\frac{5}{d}}\Eip{-\frac{6}{d}} g_2(d)}& \\
& = \bigg( \int_{\frac{1}{d}}^\infty \frac{e^{-t}}{t}dt\bigg)\Bigg[\int_{\frac{1}{d}}^\infty \frac{2e^{-t}}{t} + \frac{3e^{-t}}{t+\frac{5}{d}} - \frac{8e^{-t}}{t+\frac{2}{d}}dt\bigg]
+ 3\bigg(\int_{\frac{1}{d}}^\infty \frac{e^{-t}}{t+\frac{5}{d}}dt\bigg)\bigg(\int_{\frac{1}{d}}^\infty \frac{e^{-t}}{t+\frac{2}{d}}dt\bigg)\\
&=\bigg( \int_{\frac{1}{d}}^\infty \frac{e^{-t}}{t}dt\bigg)\bigg[\int_{\frac{1}{d}}^\infty \frac{(-3t^2 - \frac{20}{d}t + \frac{20}{d^2})e^{-t}}{t(t+\frac{2}{d})(t+\frac{5}{d})}dt \bigg]\\
&\hspace{1cm} + 3\bigg(\int_{\frac{1}{d}}^\infty \frac{e^{-t}}{t+\frac{5}{d}}dt\bigg)\Bigg[\bigg( \int_{\frac{1}{d}}^\infty \frac{e^{-t}}{t}dt\bigg) - \bigg( \int_{\frac{1}{d}}^\infty \frac{\frac{2}{d}e^{-t}}{t(t+\frac{2}{d})}dt\bigg)\Bigg]\\
&= \bigg( \int_{\frac{1}{d}}^\infty \frac{e^{-t}}{t}dt\bigg) \bigg[\int_{\frac{1}{d}}^\infty \frac{(-\frac{14}{d}t + \frac{20}{d^2})e^{-t}}{t(t+\frac{2}{d})(t+\frac{5}{d})}dt \bigg]
-\frac{6}{d}\bigg( \int_{\frac{1}{d}}^\infty \frac{e^{-t}}{t+\frac{5}{d}}dt\bigg)\bigg( \int_{\frac{1}{d}}^\infty \frac{e^{-t}}{t(t+\frac{2}{d})}dt\bigg)\\ 
&<\bigg( \int_{\frac{1}{d}}^\infty \frac{e^{-t}}{t}dt\bigg)\bigg[\int_{\frac{1}{d}}^\infty \frac{(-\frac{14}{d}t + \frac{20}{d^2})e^{-t}}{t(t+\frac{2}{d})(t+\frac{5}{d})}dt \bigg]
-\frac{1}{d}\bigg( \int_{\frac{1}{d}}^\infty \frac{e^{-t}}{t}dt\bigg)\bigg( \int_{\frac{1}{d}}^\infty \frac{e^{-t}}{t(t+\frac{2}{d})}dt\bigg)\\
&=\bigg( \int_{\frac{1}{d}}^\infty \frac{e^{-t}}{t}dt\bigg)\bigg[\int_{\frac{1}{d}}^\infty \frac{(-\frac{15}{d}t + \frac{15}{d^2})e^{-t}}{t(t+\frac{2}{d})(t+\frac{5}{d})}dt \bigg]\\
&<0.
\end{align*}
Therefore, \eqref{eq:g_condition} holds.
%%%%%
\end{proof}

\bibliographystyle{myplainnat}
\bibliography{yss-group.bib}

\end{document}